\documentclass[11pt]{article}

\usepackage[margin=1in]{geometry}
\usepackage[T1]{fontenc}
\usepackage[utf8]{inputenc}
\usepackage{lmodern}
\usepackage{microtype}
\usepackage{amsmath,amssymb,amsfonts,amsthm}
\usepackage{mathtools}
\usepackage{booktabs}
\usepackage{array}
\usepackage{enumitem}
\usepackage{natbib}
\usepackage{hyperref}
\usepackage{xurl}
\usepackage{cleveref}
\usepackage{caption}
\usepackage{longtable}

\hypersetup{
  colorlinks=true,
  linkcolor=blue,
  citecolor=blue,
  urlcolor=blue,
  pdftitle={Toward AI-Resilient Assessment in Computer Science Courses in an AI-Native World: Formalism and Construction},
  pdfauthor={Anshumali Shrivastava}
}

\setlist[itemize]{topsep=2pt,itemsep=2pt,leftmargin=*}
\setlist[enumerate]{topsep=2pt,itemsep=2pt,leftmargin=*}

\theoremstyle{definition}
\newtheorem{definition}{Definition}

\newtheorem{designcriterion}{Design Criterion}
\theoremstyle{plain}
\newtheorem{theorem}{Theorem}
\newtheorem{proposition}{Proposition}
\newtheorem{corollary}{Corollary}
\theoremstyle{remark}
\newtheorem{remark}{Remark}

\newcommand{\R}{\mathbb{R}}
\newcommand{\PF}{\operatorname{PF}}
\newcommand{\Down}{\operatorname{Down}}
\newcommand{\HV}{\operatorname{HV}}
\newcommand{\AI}{\mathrm{AI}}
\newcommand{\BF}{\mathrm{BF}}
\newcommand{\MB}{\mathrm{MB}}
\newcommand{\GB}{\mathrm{GB}}
\newcommand{\doi}[1]{\href{https://doi.org/#1}{doi:#1}}
\newcommand{\Priv}{\mathrm{priv}}
\newcommand{\bud}{\mathrm{bud}}
\newcommand{\evalerr}{\mathrm{eval}}
\newcommand{\srv}{\mathrm{srv}}
\newcommand{\hid}{\mathrm{hid}}
\newcommand{\pospart}[1]{\left[#1\right]_+}

\title{\vspace{-1.0em}Toward AI-Resilient Assessment in Computer Science Courses in an AI-Native World\\\large Formalism and Construction}
\author{Anshumali Shrivastava\\\texttt{anshumali@rice.edu}\\Department of Computer Science, Rice University}
\date{Position paper, 15th June 2026}

\begin{document}
\maketitle

\section*{Preface and Positioning}

This paper is motivated by the author's experience during a sabbatical from Rice University, spent as a full-time software engineer close to frontier AI development at Meta Superintelligence Labs. From that vantage point, the rapid transformation of software engineering practice made clear that senior computer science courses must be refined for an AI-native world.

That refinement should not be driven only by speculation, anxiety about cheating, or analogies to earlier technologies. It should begin by making the goals of a course explicit. The purpose of this paper is to move the discussion toward concrete objects that can be inspected and debated.

This paper is an attempt to formalize the ideas of AI-resilient skills and their assessments and present a concrete assignment and evaluation design for community feedback. The proposal focuses on senior undergraduate and graduate computer science coursework, where students are expected to combine algorithmic understanding, implementation skill, and systems-level judgment. The paper argues that once the goals of an AI-native assignment are made explicit, AI-resilient evaluation becomes a design problem.

The proposal is preliminary and does not claim to solve all problems in AI-native education. Several limitations are discussed in Section~\ref{sec:discussion}. The aim is to put forward a concrete framework for critique, refinement, and debate before a planned Fall 2026 deployment in COMP 480/580 at Rice University, where the author is the course instructor.

\begin{abstract}
AI-native course assessments in senior computer science courses and related fields should grade students by \emph{AI-resilient skill}: the ability to achieve outcomes beyond a strong AI baseline. Such assessments should allow students to use AI freely, while reducing the extent to which greater private AI budget or more intensive AI use, by itself, becomes a grading advantage. This paper proposes a minimal formal framework for this goal. The framework specifies a real task, an executable evaluator, a declared AI-native Pareto frontier, and a grading rule based on Pareto surplus. The central claim is simple: Pareto surplus provides a measurable, protocol-relative certificate that a submitted artifact achieves a tradeoff not already supplied by the declared AI baseline, and grading by this surplus is AI-resilient with respect to that baseline. Interpreting surplus as evidence of student skill requires the surrounding assessment protocol--for example, design reports, ablations, prompt traces, oral checks, or reproducibility explanations--but the grading certificate itself is behavioral and executable. The framework is then extended to practical complications, including self-improving AI loops, budget neutrality, server-mediated feedback, and prompt-based red teaming. As a concrete instantiation, we describe an AI-resilient approximate-membership assignment centered on Bloom filters for COMP 480/580 at Rice University, designed to test whether students can improve beyond AI-generated implementations.
\end{abstract}

\section{Introduction}

Generative AI has changed the meaning of a programming assignment. A student can now produce code, explanations, tests, reports, and design sketches using an AI system whose output may be difficult to distinguish from human work. The immediate institutional reaction has often been to ask how to prohibit, detect, or proctor against AI use. Those questions matter, but they do not address the central educational question:
\begin{quote}
What skill is the student learning that makes them more capable than a person who merely has access to AI?
\end{quote}

Recent evidence suggests that AI use in coursework is no longer marginal. The 2026 HEPI/Kortext survey of UK undergraduates reports that 95\% of students use AI in at least one way and 94\% use generative AI to help with assessed work; it also reports that fewer than half feel teaching staff are helping them develop AI skills for future careers \citep{hepi2026}. A 2026 study in \emph{Science}, based on responses from 95,513 students at 20 major public research-intensive universities in the United States, argues for discipline-specific assessment reform rather than blanket bans or universal detection regimes \citep{chirikov2026}. Princeton's 2026 decision to require proctoring for all in-person examinations, after a long honor-code tradition, is one visible example of how AI has changed the cost of relying only on trust and post-hoc reporting \citep{williams2026princeton}.

For senior-level computer science courses, the right response is not to pretend that AI does not exist. Introductory courses may still need short no-AI checkpoints for syntax, definitions, and prerequisite reasoning. But advanced courses should also ask what students can do after AI has produced the standard implementation, standard proof sketch, and standard explanation. In that setting, AI should become the baseline, collaborator, and instrument that the student must learn to direct.

The calculator analogy is useful but incomplete. Calculators did not remove the need to understand arithmetic; they changed what technical education should emphasize. It would be wasteful for an engineering course to spend most of its effort grading long manual arithmetic while pretending calculators do not exist. Similarly, once AI can generate routine code and explanations, an advanced assignment can move student effort toward evaluator interpretation, performance diagnosis, algorithm--systems tradeoffs, experimental design, robustness, and scale. TEQSA and UNESCO guidance likewise warns against relying only on detectors, prohibition, or conventional ``AI-proof'' exams and argues for broader adaptation in assessment and pedagogy \citep{lodge2024teqsa,lodge2023calculator,unesco2023}.

The same AI tools that force a rethinking of student assessment also make better assessment design more feasible for faculty. An AI-resilient assignment requires more than a prompt and a rubric: it needs an executable evaluator, strong baselines, benchmark instances, feedback rules, and red-team attempts. This would have been expensive to construct manually for every topic. AI-assisted coding changes that cost structure. Instructors can now use AI to prototype evaluators, generate baseline solutions, stress-test hidden cases, profile implementations, and search for weaknesses in the grading rule. The faculty role therefore shifts upward as well: from manually constructing every artifact to specifying the educational objective, validating the evaluator, choosing the revealed feedback, and deciding what frontier movement should count as evidence of skill. Thus AI-native assessment is not only a burden created by AI; it is also made practical by AI.

This paper proposes a direct design rule:
\begin{quote}
A senior-level AI-native assignment should grade what a trained student can achieve with AI beyond an explicit AI-native frontier, not whether the student avoided AI.
\end{quote}

The paper starts with the smallest useful mathematical model. A real task has an executable evaluator that returns a vector of scores. AI-generated submissions form a baseline set. Their nondominated scores form an AI-native Pareto frontier. A student submission earns surplus only if it expands the dominated hypervolume beyond that frontier. Under this model, the frontier is not decoration or leaderboard gamification; it is the operational measure of whether the submitted artifact has produced a tradeoff that the declared AI baseline did not provide.

The key theorem then becomes meaningful. If grades are assigned by Pareto surplus over the AI-native frontier, then AI-native submissions receive only the AI-baseline grade, and any above-baseline grade certifies measured improvement beyond the frontier. This is the mathematical reason to use Pareto frontiers in the assignment: they turn a vague debate about cheating into a concrete question about measurable frontier expansion.

After presenting this minimal core, the paper extends the framework to the complications that real courses face: self-improving AI loops, unequal AI budgets, server-mediated feedback, hidden inputs and outputs, thresholded grading, and prompt-based red teaming. These extensions are important, but they should not obscure the central idea. The core of the paper is the claim that movement beyond an AI-native Pareto frontier is the measurable artifact-level signal around which AI-resilient skill can be assessed.

The concrete construction is an approximate-membership assignment centered on Bloom filters for COMP 480/580 at Rice University. Students are not merely asked to implement a textbook Bloom filter. They are given empirical AI-native frontiers for two regimes--MB- and GB-scale data--and asked to expand one or more frontiers under hidden evaluation. The target skill includes Bloom-filter theory, false-positive analysis, hash design, memory layout, cache behavior, construction throughput, query throughput, and AI-assisted engineering judgment.

\paragraph{Contributions.}
The paper makes five contributions.
\begin{enumerate}
    \item It gives a minimal formal model in which Pareto frontier expansion measures artifact-level performance beyond an AI baseline and can serve as evidence of AI-resilient skill under an assessment protocol.
    \item It proves that Pareto-surplus grading is AI-resilient with respect to the declared AI-native frontier.
    \item It extends the minimal model to self-improving AI loops, budget neutrality, server-mediated feedback, AI-use competence, and prompt-skill red-teaming criteria.
    \item It explains how executable surplus grading can be automated and how relative surplus grading can reduce incentives to leak frontier-moving solutions.
    \item It instantiates the framework as a two-regime Bloom-filter assignment for COMP 480/580.
\end{enumerate}

\section{Minimal Framework}
\label{sec:minimal}

This section gives the bare minimum needed to state the main claims. The extensions in \Cref{sec:extensions} specify how the AI baseline is generated, how server feedback is controlled, and how private AI-budget effects are handled. The core model only needs four objects: a task, an evaluator, an AI-native frontier, and a grading rule.

\subsection{Real tasks and score vectors}

We treat a program in the standard theory-of-computation sense as a finite effective description of a computation; equivalently, for this paper, an executable artifact may be modeled as a Turing-machine description or partial computable function \citep{sipser2012}. We do not formalize programming-language syntax further. The course-specific object is the interface and benchmark protocol that separates raw submissions, admissible submissions, and failed submissions.

\begin{definition}[Task interface and benchmark protocol]
A \emph{task interface and benchmark protocol} is a tuple
\[
\mathcal I=(\Omega_{\mathcal I},\mathrm{API},\mathcal B_{\mathrm{dev}},\mathcal B_{\mathrm{eval}},\mathcal C,\mathcal E_{\mathrm{run}},m,\mathcal N),
\]
where \(\Omega_{\mathcal I}\) is the raw submission space of executable artifacts that attempt the assignment; \(\mathrm{API}\) specifies the required program interface and allowed dependencies; \(\mathcal B_{\mathrm{dev}}\) is the public development benchmark; \(\mathcal B_{\mathrm{eval}}\) is the hidden or final evaluation benchmark; \(\mathcal C\) is the set of hard correctness constraints; \(\mathcal E_{\mathrm{run}}\) specifies hardware, compiler, flags, timeouts, memory limits, and random-seed policy; \(m=(m_1,\ldots,m_d)\) is the vector of measured raw objectives; and \(\mathcal N\) specifies a scoring map \(\nu_{\mathcal I}\), the official larger-is-better score coordinates, the bounded nonfailing score box, clipping rules, the reference point, and the fixed failing score \(f_{\mathrm{fail},\mathcal I}\). The final benchmark may be a fixed suite, a distribution over instances, or a seed-based sampler.

Once \(\mathcal I\) is fixed, it induces the admissible submission set
\[
\Pi_{\mathcal I}=\{\pi\in\Omega_{\mathcal I}: \pi \text{ satisfies } \mathrm{API}, \text{ runs under } \mathcal E_{\mathrm{run}}, \text{ and satisfies } \mathcal C \text{ on } \mathcal B_{\mathrm{eval}}\}.
\]
A raw submission outside \(\Pi_{\mathcal I}\) is nonadmissible. It may receive the fixed failing score for grading and audit logs, but it is excluded from AI-frontier construction and from positive surplus credit.
\end{definition}

\begin{definition}[Real task]
A \emph{real task} is a tuple
\[
T=(\mathcal I,E_T),
\]
where \(\mathcal I\) is a task interface and benchmark protocol, and \(E_T:\Omega_{\mathcal I}\to\R^d\) is a total executable evaluator in the official normalized score coordinates. For every admissible submission \(\pi\in\Pi_{\mathcal I}\),
\[
E_T(\pi)=\nu_{\mathcal I}(m(\pi)).
\]
For nonadmissible raw submissions, \(E_T(\pi)=f_{\mathrm{fail},\mathcal I}\). All hypervolume and Pareto-frontier claims below concern admissible, nonfailing score vectors; nonadmissible submissions receive no surplus credit.
\end{definition}

\noindent\textbf{Plain meaning.} A real task is not just a written assignment. It specifies what students may submit, how submissions are run, which hidden benchmarks are used, which constraints are mandatory, which score transformations are official, and which numerical objectives determine the score. The submitted artifact is judged by executable behavior, not by authorship.

The score order is not an arbitrary component of the task. Throughout the paper, all official score coordinates are normalized so that larger values are better, and order means coordinate domination. For score vectors \(x,y\in\R^d\), write \(x\preceq y\) to mean that \(y\) weakly dominates \(x\), i.e., \(x_i\leq y_i\) for every coordinate. Write \(x\prec y\) when \(y\) strictly dominates \(x\): \(x\preceq y\) and \(x_j<y_j\) for at least one coordinate.

\subsection{AI-native Pareto frontier}

Let \(S\subseteq\R^d\) be a finite set of score vectors. The Pareto frontier is the nondominated subset
\[
\PF(S)=\{s\in S:\nexists s'\in S\text{ such that }s\prec s'\}.
\]
This definition is standard in multi-objective optimization \citep{deb2001,miettinen1999}.

\begin{definition}[Declared AI-native baseline]
For a real task \(T\), a \emph{declared AI-native baseline} is a finite, nonempty set
\[
\Pi_{\AI}(T)\subseteq\Pi_{\mathcal I}
\]
of admissible submissions frozen before grading. Its elements are either produced by the stated baseline-generation process or are explicitly declared admissible fallback submissions. It is not the set of all programs that any present or future AI system could ever produce. If no AI-generated candidate is admissible after applying the hard constraints, the assignment must include at least one declared fallback submission in \(\Pi_{\mathcal I}\), and the baseline-freeze record must identify it as a fallback rather than as an AI-generated result.
\end{definition}

The AI-native score set is
\[
S_{\AI}(T)=\{E_T(\pi):\pi\in\Pi_{\AI}(T)\},
\]
and the AI-native Pareto frontier is
\[
F_{\AI}(T)=\PF(S_{\AI}(T)).
\]

\noindent\textbf{Plain meaning.} The AI-native frontier is the declared set of best-known AI-generated tradeoffs for the assignment. It is the line students must move to receive surplus credit.

\subsection{Pareto surplus}

Before the assignment is released, the task fixes the official score coordinates, all normalizations or monotone transformations, a bounded nonfailing score box, and a reference point \(r\in\R^d\) strictly dominated by every nonfailing score in that box. Scores outside the box are either clipped to the box or assigned the fixed failing score according to the public evaluator rule. The failing score is not used as a frontier point. The reference point is not chosen after seeing student submissions. Hypervolume surplus is computed only in these official normalized coordinates; different transformations can induce different hypervolume values, so the transformations are part of the task definition.

Throughout the formal grading claims, every set passed to \(\HV\) is finite. For a finite set \(F\subset\R^d\), define the dominated hypervolume
\[
\HV(F;r)=\lambda_d\left(\bigcup_{s\in F}[r_1,s_1]\times\cdots\times[r_d,s_d]\right),
\]
where \(\lambda_d\) is \(d\)-dimensional Lebesgue measure. Hypervolume is a standard indicator for multi-objective performance \citep{zitzler1999,zitzler2003}.

For an admissible submission \(\pi\in\Pi_{\mathcal I}\), its Pareto surplus over the AI-native frontier is
\[
\delta_T(\pi)
=
\HV\left(F_{\AI}(T)\cup\{E_T(\pi)\};r\right)
-
\HV\left(F_{\AI}(T);r\right).
\]
For a nonadmissible raw submission, define \(\delta_T(\pi)=0\) for the surplus component; the complete grade may still assign a failing score outside the surplus component.

\noindent\textbf{Plain meaning.} Pareto surplus is the new tradeoff volume added by the submitted artifact. If \(\delta_T(\pi)=0\), the artifact did not move the frontier. If \(\delta_T(\pi)>0\), the artifact achieved a measured tradeoff not already covered by the declared AI baseline.

\subsection{AI-resilient skill}

Let \(\sigma\) be a candidate human skill. In the Bloom-filter example, \(\sigma\) may include probabilistic analysis, bits-per-key tradeoffs, hash-function design, cache-aware layout, vectorization, parallel construction, and AI-guided engineering judgment. The paper does not try to model cognition directly. It models the finite evidence produced by an assessment protocol.

Let \(\Gamma\) denote a specified finite assessment protocol: for example, a finite set of students, time limits, permitted AI use, required reports, oral checks, and the evaluator runs used to produce scores. For a skill \(\sigma\), let \(S_H^\Gamma(T;\sigma)\) be the finite set of score vectors produced by that protocol for students or humans possessing \(\sigma\). Define
\[
\Delta_T^\Gamma(\sigma)
=
\HV\left(F_{\AI}(T)\cup S_H^\Gamma(T;\sigma);r\right)
-
\HV\left(F_{\AI}(T);r\right).
\]

\begin{definition}[AI-resilient human skill]
A skill \(\sigma\) is \emph{AI-resilient for task \(T\) under assessment protocol \(\Gamma\)} if \(\Delta_T^\Gamma(\sigma)>0\). It is \(\varepsilon\)-AI-resilient if \(\Delta_T^\Gamma(\sigma)\geq\varepsilon\) for a pre-specified margin \(\varepsilon>0\).
\end{definition}

\noindent\textbf{Plain meaning.} A skill is AI-resilient when the assessment protocol produces artifacts that move the measured frontier. Pareto surplus is first an artifact-level certificate. Treating it as evidence of student skill is strongest when paired with supporting evidence such as design reports, ablations, prompt traces, oral checks, or reproducibility explanations.

\subsection{AI-resilient grading}

\begin{definition}[AI-resilient surplus grading]
Let \(g_{\AI}\in\R\) be the declared AI-baseline surplus grade. A surplus grading component \(G_T:\Omega_{\mathcal I}\to\R\) is \emph{AI-resilient with respect to the declared AI frontier} if nonadmissible submissions receive no above-baseline surplus credit, every AI-native submission receives at most \(g_{\AI}\), and any surplus grade above \(g_{\AI}\) requires an admissible submission with positive Pareto surplus:
\[
\pi\notin\Pi_{\mathcal I}\implies G_T(\pi)\leq g_{\AI},
\qquad
\pi\in\Pi_{\AI}(T)\implies G_T(\pi)\leq g_{\AI},
\]
\[
G_T(\pi)>g_{\AI}\implies \pi\in\Pi_{\mathcal I}\text{ and }\delta_T(\pi)>0.
\]
\end{definition}

\noindent\textbf{Plain meaning.} The surplus component of the grade does not ask who typed the code. It asks whether the submitted system does something measurably better than the declared AI baseline. A complete course grade may also include correctness, reproducibility, report quality, or minimum-performance requirements; zero surplus should not be read as matching the AI frontier unless those baseline-level requirements are also met.

\section{Core Claims}
\label{sec:core-claims}

The definitions above are useful only if they imply concrete grading statements. The results in this section are conditional design guarantees: once the task, coordinate score order, score transformations, reference point, AI baseline, feedback policy, and grading rule are fixed, the following statements hold. They do not claim that the declared baseline exhausts all possible AI systems, nor that artifact-level surplus alone identifies the causal source of the improvement.

\begin{theorem}[Pareto surplus characterizes frontier expansion]
\label{thm:skill-measure}
For a real task \(T\), let \(F_{\AI}(T)\) be the declared AI-native frontier, and assume scores lie in the fixed nonfailing score box strictly above the reference point \(r\). For any admissible submission \(\pi\),
\[
\delta_T(\pi)>0
\quad\Longleftrightarrow\quad
E_T(\pi)\notin\Down(F_{\AI}(T)),
\]
where \(\Down(F)=\{x:\exists y\in F\text{ with }x\preceq y\}\). Consequently, for a candidate skill \(\sigma\) evaluated through a finite assessment protocol,
\[
\Delta_T^\Gamma(\sigma)>0
\quad\Longleftrightarrow\quad
\exists s\in S_H^\Gamma(T;\sigma)\text{ such that }s\notin\Down(F_{\AI}(T)).
\]
\end{theorem}

\noindent\textbf{Practical implication.} Pareto surplus is a measurable, artifact-level certificate of frontier expansion. If trained students consistently produce positive surplus, the course has evidence that the target skill creates outcomes beyond the declared AI baseline. If they cannot, then the skill may not be AI-resilient for this task, the baseline may already capture the skill, or the evaluator may not measure the skill. The proof is in \Cref{app:proof-skill}.

\begin{corollary}[Null-frontier outcome]
\label{cor:null-frontier}
For a real task \(T\), declared AI-native frontier \(F_{\AI}(T)\), and finite assessment score set \(S_H^\Gamma(T;\sigma)\), suppose
\[
\Delta_T^\Gamma(\sigma)=0.
\]
Then the assessment protocol has not witnessed \(\sigma\) as an AI-resilient skill for \(T\). Equivalently, every score vector produced by the finite assessment protocol is dominated by the declared AI-native frontier.
\end{corollary}

\noindent\textbf{Practical implication.}
A \emph{null-frontier outcome} is a diagnostic signal: under this task, evaluator, baseline, and assessment protocol, the proposed skill did not produce measurable frontier movement. This does not prove that the skill is useless. It may mean that the evaluator is too narrow, the baseline is already strong, the assignment did not expose the intended skill, or students were not given enough support. But it does raise the curricular question of whether this skill should be the central grade-bearing target of a senior AI-native assignment.

Define the surplus component of the grade by
\[
G_T(\pi)=g_{\AI}+\phi(\delta_T(\pi)),
\]
where \(\phi:\R_{\geq0}\to\R_{\geq0}\) is monotone and \(\phi(0)=0\), and where nonadmissible submissions have \(\delta_T(\pi)=0\) for this surplus component. The map \(\phi\) is part of the official grading policy. It may be an absolute scale, a percentile-calibrated scale over final surplus values, or a hybrid. In all cases, the input to \(\phi\) is the evaluator-computed surplus, not the claimed authorship of the submission.

\begin{theorem}[Pareto-surplus grading is AI-resilient]
\label{thm:pareto-grading}
Let \(T\) be a real task and let \(F_{\AI}(T)=\PF(S_{\AI}(T))\) be the declared AI-native frontier. Under Pareto-surplus grading:
\begin{enumerate}
    \item every AI-native submission in \(\Pi_{\AI}(T)\) receives surplus grade exactly \(g_{\AI}\);
    \item every submission receiving surplus grade above \(g_{\AI}\) has positive Pareto surplus beyond \(F_{\AI}(T)\).
\end{enumerate}
\end{theorem}

\noindent\textbf{Practical implication.} Simply giving the assignment to the baseline AI process earns the baseline surplus grade, because that behavior is already represented in the frontier. Above-baseline surplus credit starts only when the submitted artifact moves the measured frontier. The proof is in \Cref{app:proof-grading}.

\begin{remark}[Authorship detection is not the grading mechanism]
Under Pareto-surplus grading, the surplus grade depends on executable behavior relative to the published frontier. The course need not infer whether a submission was written by a human, by AI, or by both. Ordinary rules about collaboration, attribution, safety, and unauthorized access to hidden benchmarks still apply.
\end{remark}

\section{Extensions for Real Course Settings}
\label{sec:extensions}

The minimal framework establishes the measurement and grading principle. A real course needs additional machinery to make the frontier meaningful and the grading fair. This section adds that machinery without changing the core claim.

\subsection{AI-agent protocols and self-improving loops}

The phrase ``what AI can do'' is not a mathematical object. The course must specify a baseline-generation process. An AI-agent protocol is therefore a reproducible trace-generating procedure, not merely a list of model names.

\begin{definition}[AI-agent protocol]
An \emph{AI-agent protocol} for task \(T\) consists of metadata
\[
(\mathcal M,\mathcal P,b,\mathcal O,\mathsf{Sel})
\]
and an executable or auditable run procedure. Here \(\mathcal M\) is the class of models or agents, \(\mathcal P\) is the prompting and interaction protocol, \(b\) is a resource budget with a declared resource accounting rule, \(\mathcal O\) is the allowed oracle access, and \(\mathsf{Sel}\) is the declared candidate-selection rule. When run under the declared seeds or randomness, the protocol produces a finite trace \(\tau_{\mathcal A}\) of prompts, model outputs, tool calls, oracle responses, resource use, and candidate programs, with total resource use at most \(b\). The selected candidate set \(\Pi_{\mathcal A}(T)=\mathsf{Sel}(\tau_{\mathcal A})\subseteq\Omega_{\mathcal I}\) is finite.
\end{definition}

A reproducible baseline should include a baseline-freeze record: model names and versions, prompts and system instructions, tools and libraries, public and hidden feedback access, server-query budget, random seeds where applicable, compute or wall-clock budget, number of generated candidates, final filtering or selection procedure, and the date and time at which the frontier was frozen. These details are not theoretical decoration; they make the declared frontier auditable.

Modern AI systems can self-improve within a task by generating a candidate, observing feedback, revising the solution, and repeating. Therefore, if the assignment allows students to run such loops privately, the baseline should include them.

\begin{definition}[Looped AI protocol]
A \emph{looped AI protocol} is an AI-agent protocol whose trace contains a sequence
\[
\pi^{(0)},h^{(0)},\pi^{(1)},h^{(1)},\ldots,\pi^{(t)},
\]
where \(\pi^{(i)}\) is a candidate program, \(h^{(i)}\) is the feedback available under the oracle policy, each later candidate is generated from the previous trace according to the declared interaction protocol, and the total resource use of the sequence is at most the declared budget \(b\). The output is the finite set chosen from this trace by \(\mathsf{Sel}\).
\end{definition}

\noindent\textbf{Design rule.} Do not compare students against a one-shot AI baseline if students can use looped AI. Generate the baseline with looped AI under the same official feedback policy. A student should not receive above-baseline credit for rediscovering a generic generate--evaluate--revise loop that the course could have run before grading.

\subsection{Server-mediated feedback and hidden evaluation}

The assignment should not monitor private AI use. It should control the official information channel.

\begin{definition}[Server-mediated feedback oracle]
A \emph{server-mediated feedback oracle} for task \(T\) is a tuple
\[
\mathcal O_{\srv}=(\mathcal B_{\hid},q,E_{T,\hid},\psi),
\]
where \(\mathcal B_{\hid}\) is a hidden development benchmark specified by the interface or server policy, \(q\) is the official query budget, \(E_{T,\hid}:\Omega_{\mathcal I}\to\R^d\) is the server-side development evaluator in the same official score coordinates as \(E_T\), and \(\psi:\R^d\to\mathcal H\) maps full hidden-development scores to a restricted feedback space \(\mathcal H\). On a submitted program \(\pi\), the server computes \(E_{T,\hid}(\pi)\) but releases only
\[
\psi(E_{T,\hid}(\pi)).
\]
The final grading evaluator remains \(E_T\) on the final evaluation benchmark; the feedback oracle is an official information channel, not a redefinition of \(E_T\). The map \(\psi\) may reveal aggregate pass/fail status, coarse score bins, distance-to-frontier summaries, or leaderboard position. It should not reveal hidden inputs, hidden outputs, per-instance labels, or per-instance errors.
\end{definition}

\noindent\textbf{Design rule.} Students may use AI freely between server queries, but every student receives the same official feedback budget and the same type of feedback. The course controls benchmark information, not private thought or private AI use.

\subsection{Budget neutrality margin through thresholded surplus}

Budget neutrality is not an absolute claim that all students have identical private resources, or that no amount of additional AI use can ever help. It is a protocol-relative release condition. The assignment declares an envelope of plausible private AI-only protocols, restricts the official feedback channel, freezes an AI-native frontier, and sets a threshold large enough to absorb observed AI-only frontier movement and evaluator noise. The resulting claim is therefore conditional: within the declared envelope and feedback policy, private AI budget alone should not reliably produce grade-bearing surplus.

Students with more private AI budget may generate more candidates between server queries. The design goal is not to make AI-use skill worthless; it is to make private purchasing power insufficient, by itself, for grade-bearing surplus within the declared feedback policy and threshold. Let \(F_0\) be the public frontier used for grading, let \(\mathfrak A_{\Priv}\) be a declared envelope of plausible private AI-only protocols that receive no more than the official server feedback, and let \(E_T^*(\pi)\) denote the ideal deterministic score used for saturation analysis

The following is an idealized saturation assumption. It is useful because it states exactly what the threshold theorem requires, but in deployment it must be approximated by finite red-team experiments and conservative safety margins.

For a protocol \(\mathcal A\), let \(\Pi^*(\mathcal A)\subseteq\Pi_{\mathcal I}\) be the finite set of admissible submissions returned by its declared selection rule under ideal evaluation, and define
\[
S^*(\mathcal A)=\{E_T^*(\pi):\pi\in\Pi^*(\mathcal A)\}.
\]

Thus, when an AI-only protocol \(\mathcal A\) produces a submitted artifact \(\pi\) under this envelope, the model assumption is \(\pi\in\Pi^*(\mathcal A)\), and therefore \(E_T^*(\pi)\in S^*(\mathcal A)\). We say \(F_0\) is \(\varepsilon_{\bud}\)-saturated against the declared private AI-only envelope when
\[
\sup_{\mathcal A\in\mathfrak A_{\Priv}}
\left[
\HV(F_0\cup S^*(\mathcal A);r)-\HV(F_0;r)
\right]
\leq \varepsilon_{\bud}.
\]

In an actual course deployment, the true supremum over \(\mathfrak A_{\Priv}\) is not observable. The instructor instead runs a finite red-team suite
\(\mathfrak A_{\mathrm{red}}\subseteq\mathfrak A_{\Priv}\) before release and measures
\[
\varepsilon_{\bud}^{\mathrm{red}}
=
\max_{\mathcal A\in\mathfrak A_{\mathrm{red}}}
\left[
\HV(F_0\cup S^*(\mathcal A);r)-\HV(F_0;r)
\right].
\]
This red-team value is not a second grading threshold. It is used to choose the declared private-AI-budget margin
\[
\varepsilon_{\bud}
=
\varepsilon_{\bud}^{\mathrm{red}}
+
\gamma_{\mathrm{safety}},
\]
where \(\gamma_{\mathrm{safety}}\geq 0\) accounts for untested private AI-only strategies, model variation, prompt variation, and uncertainty in the red-team suite. The theorem below is conditional on this declared margin satisfying the saturation inequality above. If later red-team attempts move the frontier by more than \(\varepsilon_{\bud}\), the assignment has not achieved budget neutrality at the chosen margin; the frontier, feedback policy, or declared margin must be revised.

The evaluator may have measurement noise from timing variance, sampling variance in false-positive estimates, or randomized workloads. Let \(E_T^*(\pi)\) denote the ideal deterministic score and let \(\widehat E_T(\pi)\) denote the observed score. Define
\[
\delta_T^*(\pi;F_0)
=
\HV(F_0\cup\{E_T^*(\pi)\};r)-\HV(F_0;r),
\]
and define \(\widehat\delta_T(\pi;F_0)\) analogously using \(\widehat E_T(\pi)\). Assume the evaluator satisfies the one-sided hypervolume error bound
\[
\widehat\delta_T(\pi;F_0)
\leq
\delta_T^*(\pi;F_0)+\varepsilon_{\evalerr}
\]
uniformly for every admissible submission considered in grading, at the stated confidence level if the evaluator is randomized. The single deployed grading threshold is
\[
\eta_T
=
\varepsilon_{\bud}
+
\varepsilon_{\evalerr}.
\]

Thresholded grading uses the observed surplus:
\[
G_T^{\eta}(\pi)
=
 g_{\AI}+\phi\left(\pospart{\widehat\delta_T(\pi;F_0)-\eta_T}\right).
\]

\begin{theorem}[Thresholded grading separates surplus from budget and noise]
\label{thm:thresholded}
Assume \(F_0\) is \(\varepsilon_{\bud}\)-saturated against the declared private AI-only envelope \(\mathfrak A_{\Priv}\), and the evaluator satisfies the uniform one-sided hypervolume error bound above. Under thresholded grading, with the stated confidence level for the evaluator-noise bound, any AI-only submission inside the declared envelope receives no surplus grade above \(g_{\AI}\). Moreover, whenever the one-sided evaluator-error bound holds, any admissible submission with \(G_T^{\eta}(\pi)>g_{\AI}\) has ideal surplus
\[
\delta_T^*(\pi;F_0)>\varepsilon_{\bud}.
\]
\end{theorem}

\noindent\textbf{Practical implication.} The threshold is the fairness buffer. Small improvements that could be explained by private AI-only search or measurement noise do not become surplus grades. Above-baseline surplus credit starts only after the artifact clears the declared private-AI-budget margin in the ideal evaluator, up to the stated evaluator-noise confidence level. The proof is in \Cref{app:proof-thresholded}.

\subsection{AI-use competence is part of the target}

Budget neutrality should not be confused with neutrality to AI skill. A student who can use AI well should often do better. Useful AI competence includes decomposing the task, writing diagnostic prompts, generating implementation variants, checking generated claims, using server feedback to form hypotheses, and combining model output with domain knowledge. In an AI-native course, this is not a loophole; it is one of the intended capabilities.

The grading design therefore separates two things:
\[
\text{productive AI-use competence} \quad \neq \quad \text{unequal private AI budget}.
\]
The former should be rewarded. The latter should be limited by hidden benchmarks, fixed server feedback, equal official query budgets, and thresholded surplus.

\section{Concrete Construction: AI-Resilient Bloom Filter Assignments}
\label{sec:construction}

We now instantiate the framework in a senior-level computer science assignment. The target skill is not merely the ability to state the definition of a Bloom filter or reproduce its standard implementation. The intended skill is the ability to use probabilistic data-structure knowledge and systems judgment to improve a performance frontier. This skill includes approximate set membership, false-positive analysis, hash-function design, memory-layout choices, insertion/query-time tradeoffs, cache behavior, parallelism, and empirical performance measurement.

The construction is centered on Bloom's classical approximate-membership structure, which trades space for a controlled false-positive probability while preserving no false negatives \citep{bloom1970}. Bloom filters and their variants are standard tools in databases, networks, storage systems, and distributed systems \citep{broder2004}. The assignment name can emphasize Bloom filters, but the admissible design space must be explicit: the course may either restrict submissions to Bloom-filter-family implementations or allow broader approximate-membership systems, such as blocked Bloom filters, quotient filters, cuckoo filters, xor/ribbon filters, or learned-filter variants. The evaluator only knows the interface; the handout must state which design family is allowed.

\subsection{Base task}

Let \(U\) be a universe of keys, and let \(S\subset U\) be a set of inserted elements with \(|S|=n\). A submitted program \(\pi\) implements
\[
\operatorname{insert}_{\pi}(x)
\quad\text{and}\quad
\operatorname{query}_{\pi}(x).
\]
The hard correctness constraint is no false negatives on the inserted set specified by the benchmark protocol:
\[
\forall x\in S,\qquad \operatorname{query}_{\pi}(x)=1.
\]
For each scale \(\lambda\), the handout declares a workload family
\[
\{\mathcal D_{\mathrm{neg}}^\lambda(\theta):\theta\in\Theta_\lambda\}.
\]
The final evaluator uses one hidden draw \(\theta_\lambda\) and sets
\[
\mathcal D_{\mathrm{neg}}^\lambda
=
\mathcal D_{\mathrm{neg}}^\lambda(\theta_\lambda).
\]
Students know the workload family, but not the hidden seed, split, mixture parameter, or final evaluation instances. The reported false-positive rate is an empirical estimate, or a high-confidence upper bound, from a declared number of nonmember queries:
\[
p_{\mathrm{fp},\lambda}(\pi)=\Pr_{z\sim\mathcal D_{\mathrm{neg}}^\lambda}[\operatorname{query}_{\pi}(z)=1]
\]
in the idealized definition, and its sampled estimate in the executable evaluator.

\paragraph{Operational evaluator details.}
For the assignment to be reproducible, the handout must fix the nonmember distribution, sample sizes for false-positive estimation, whether false-negative correctness is benchmark-only or high-confidence, the repeated-run methodology for timing, hardware, compiler, flags, thread pinning, cache warmup, randomized workload order, and the rule for computing \(\varepsilon_{\evalerr}\) in hypervolume units. These details are engineering details, but they are part of the mathematical object because they determine \(E_T\).

\subsection{Two scale regimes}

Bloom-filter skill is scale-dependent. At MB scale, a student may win through analytic parameter tuning and low-overhead implementation. At GB scale, the frontier is often determined by memory bandwidth, cache locality, parallel construction, vectorized queries, and measurement discipline. The assignment therefore uses two scale-specific tasks:
\[
\mathcal T_{\BF}=\{T_{\MB},T_{\GB}\}.
\]
Each scale \(\lambda\in\{\MB,\GB\}\) has its own evaluator \(E_\lambda\), AI-agent protocol \(\mathcal A_\lambda\), empirical AI-native frontier \(\widehat F_{\AI}^{\lambda}\), reference point \(r_\lambda\), and threshold \(\eta_\lambda\). Separate frontiers are necessary because a design excellent at MB-scale latency may be poor under GB-scale memory-bandwidth constraints.

\paragraph{MB-scale approximate membership.}
In the MB-scale regime, the raw key data is on the order of megabytes. Let \(B_{\MB}\) be the raw bytes of the key set. The cost vector is
\[
C_{\MB}(\pi)=\left(\rho_{\MB}(\pi),p_{\mathrm{fp},\MB}(\pi),\tau_{\mathrm{build},\MB}(\pi),\tau_{\mathrm{qry},\MB}(\pi)\right),
\]
where \(\rho_{\MB}(\pi)=M_{\MB}(\pi)/B_{\MB}\) is the compression ratio relative to the raw key data, and \(E_{\MB}(\pi)=-C_{\MB}(\pi)\). Skills that may move this frontier include choosing bits per key, choosing the number of hash functions, reducing hashing constants, avoiding unnecessary allocations, and tuning bit-array layout.

\paragraph{GB-scale approximate membership.}
In the GB-scale regime, the raw key data is large enough that memory hierarchy and parallelism become central. Let \(B_{\GB}\) be the raw bytes of the key set. The cost vector is
\[
C_{\GB}(\pi)=\left(\rho_{\GB}(\pi),p_{\mathrm{fp},\GB}(\pi),\frac{1}{\Theta_{\mathrm{build},\GB}(\pi)},\frac{1}{\Theta_{\mathrm{qry},\GB}(\pi)},L_{\mathrm{scale},\GB}(\pi)\right),
\]
where \(\Theta_{\mathrm{build},\GB}\) and \(\Theta_{\mathrm{qry},\GB}\) are build and query throughput. To make the scale-loss term concrete, let \(\Theta_{\mathrm{qry},\GB}(k;\pi)\) denote query throughput using \(k\) worker threads and let \(k_{\max}\) be the evaluator's declared thread count. One possible evaluator-defined metric is
\[
L_{\mathrm{scale},\GB}(\pi)=
\left[
1-
\frac{\Theta_{\mathrm{qry},\GB}(k_{\max};\pi)}{k_{\max}\Theta_{\mathrm{qry},\GB}(1;\pi)}
\right]_+.
\]
If \(\Theta_{\mathrm{qry},\GB}(1;\pi)=0\), the evaluator assigns the fixed failing score. This loss is small when parallel query throughput is close to linear and does not penalize superlinear scaling. Again \(E_{\GB}(\pi)=-C_{\GB}(\pi)\). Skills that may move this frontier include blocked Bloom filters, cache-aware layout, SIMD vectorization, batched queries, parallel construction, NUMA-aware partitioning, and memory-bandwidth measurement.

\begin{table}[t]
\centering
\caption{Two fine-grained approximate-membership scenarios centered on Bloom filters. Each scale has its own evaluator and AI-native Pareto frontier.}
\footnotesize
\setlength{\tabcolsep}{4pt}
\begin{tabular}{p{0.10\linewidth}p{0.27\linewidth}p{0.28\linewidth}p{0.25\linewidth}}
\toprule
Scale & Primary objective & Representative cost vector & Human skill emphasized \\
\midrule
MB & Compact in-memory approximate membership with low query latency & \(\rho, p_{\mathrm{fp}}, \tau_{\mathrm{build}}, \tau_{\mathrm{qry}}\) & Parameter tuning, false-positive analysis, low-overhead hashing, memory layout \\
\addlinespace
GB & High-throughput approximate membership under memory-bandwidth and parallelism constraints & \(\rho, p_{\mathrm{fp}}, 1/\Theta_{\mathrm{build}}, 1/\Theta_{\mathrm{qry}}, L_{\mathrm{scale}}\) & Cache-aware blocking, SIMD, batching, parallel construction, NUMA-aware design \\
\bottomrule
\end{tabular}
\end{table}

\subsection{Scale-specific grading}

Here \(\eta_\lambda\) is the scale-specific version of \(\eta_T\), calibrated as
\[
\eta_\lambda
=
\varepsilon_{\bud}^{\lambda}
+
\varepsilon_{\evalerr}^{\lambda}.
\]

For scale \(\lambda\in\{\MB,\GB\}\), define
\[
\delta_\lambda(\pi_\lambda)
=
\HV\left(\widehat F_{\AI}^{\lambda}\cup\{E_\lambda(\pi_\lambda)\};r_\lambda\right)
-
\HV\left(\widehat F_{\AI}^{\lambda};r_\lambda\right).
\]

The surplus grade for that scale is
\[
G_\lambda(\pi_\lambda)=g_{\AI}^{\lambda}+\phi_\lambda\left(\pospart{\delta_\lambda(\pi_\lambda)-\eta_\lambda}\right),
\]
where \(\phi_\lambda(0)=0\) and \(\phi_\lambda\) is monotone. The aggregate suite grade may be
\[
G_\Lambda(\pi_{\MB},\pi_{\GB})=\sum_{\lambda\in\{\MB,\GB\}}w_\lambda G_\lambda(\pi_\lambda),
\qquad w_\lambda\geq0,
\qquad \sum_\lambda w_\lambda=1.
\]
An \emph{AI-native suite} is a pair \((\pi_{\MB},\pi_{\GB})\) with \(\pi_\lambda\in\Pi_{\AI}(T_\lambda)\) for each scale \(\lambda\) that receives positive weight. The aggregate AI-baseline grade is the corresponding weighted aggregate of the scale baseline grades.

\begin{proposition}[Multi-frontier grading]
\label{prop:multi-frontier}
If each scale grade \(G_\lambda\) is assigned by thresholded Pareto surplus over its fixed frontier \(\widehat F_{\AI}^{\lambda}\), then an AI-native suite receives the aggregate AI-baseline grade. If the aggregate grade is above baseline, then there exists at least one scale \(\lambda\) with \(w_\lambda>0\) whose component grade has threshold-exceeding Pareto surplus.
\end{proposition}

\noindent\textbf{Practical implication.} MB-scale and GB-scale approximate-membership work emphasize different skills. Separate frontiers prevent a solution optimized for one scale from being treated as universally good. The proof is in \Cref{app:proof-multi}.

If both the MB- and GB-scale regimes produce a null-frontier outcome, that is itself informative: the construction has not witnessed approximate-membership knowledge as a grade-bearing AI-resilient skill under the declared evaluator and baseline, and the instructor should revise either the evaluator or the target capability.

\subsection{Hidden server feedback}

For each scale, the course server keeps hidden benchmark data \(\mathcal D_{\hid}^{\lambda}\) generated from the declared workload family using unreleased seeds, splits, mixture parameters, and evaluation instances. During development it releases only fixed aggregate feedback
\[
\psi_\lambda(E_{\lambda,\hid}(\pi)),
\]
such as no-false-negative pass/fail, coarse bins for memory ratio, false-positive rate, build throughput, query throughput, and frontier position. It does not release hidden inputs, hidden outputs, per-query labels, or per-instance errors. Each student receives the same official query budget \(q_\lambda\), and final grading uses a separate held-out evaluation split through \(E_\lambda\). Students may use AI freely between server queries.

If neither scale produces meaningful student surplus beyond the declared AI frontier, that outcome is itself informative: the assignment, as designed, has not witnessed Bloom-filter knowledge as a grade-bearing AI-resilient skill, and the instructor should either revise the evaluator or move the grade-bearing target to a higher-level capability built on top of approximate-membership understanding.

\paragraph{Budget-neutrality margin for the Bloom-filter assignment.}
The Bloom-filter construction does not prove budget neutrality from the mathematics of Bloom filters alone. It pursues a budget-neutrality margin operationally. For each scale \(\lambda\), the course freezes a strong AI-native frontier, red-teams generic AI-only improvement strategies, restricts official server feedback to coarse aggregate summaries, enforces the same query budget for all students, and grades only surplus beyond the published threshold \(\eta_\lambda\). The hidden nonmember workload is drawn from a declared workload family, but the final seed, split, and evaluation instances are not released. Thus students know the kind of workload for which they are designing, while the feedback channel is too limited to serve as a rich training set for benchmark overfitting.

This claim is intentionally modest. Additional private AI calls may help students search over generic implementation variants, but those variants should be represented in the frozen baseline or absorbed into \(\eta_\lambda\). If a private AI-only red team can use the allowed feedback to infer hidden workload structure and obtain threshold-exceeding surplus, then the Bloom-filter assignment has failed the budget-neutrality test at that threshold. The remedy is to regenerate the frontier, restrict feedback further, increase the safety buffer, or revise the workload family.

\subsection{Automated relative grading and solution-leakage incentives}
\label{sec:automated-incentives}

Executable evaluation also changes the incentives around collusion and unauthorized solution sharing. In a conventional assignment, copying a solution may help both the giver and the receiver if grading is mostly absolute and the copied artifact is hard to detect. In a frontier-based assignment with relative or percentile-calibrated surplus grading, a frontier-moving idea is valuable because it is scarce. If a student gives away the implementation or the key design insight that moves the frontier, the recipient can occupy the same region of surplus space. Under a percentile-calibrated surplus map, this can lower the helper's relative position and therefore reduce the helper's own grade. The mechanism does not replace collaboration rules or academic integrity policy, but it changes the incentive: leaking the final frontier-moving trick is not free.

Minimal server feedback reinforces this effect. Because the server does not release hidden inputs, hidden outputs, per-instance labels, or per-instance errors, pooling server responses is not the same as pooling a rich training set for overfitting the benchmark. Students may still discuss concepts, design principles, and debugging strategies as permitted by the course policy. But undisclosed transfer of code, prompts, parameter choices, or final optimizations that create surplus can directly help competitors in the same relative grade space.

The same structure makes grading highly automatable. Once the evaluator, frontiers, thresholds, and surplus-to-grade map are fixed, the course server can run submissions, compute score vectors, check hard constraints, compute Pareto surplus, apply percentile or absolute grade maps, and produce audit logs. Human instructors still choose the task, validate the evaluator, inspect reports or ablations, and decide the collaboration policy. But the repetitive grading work is executable. This is a case where AI and automation help instructors at least as much as they help students: AI-assisted development can build the evaluator and baselines, while the grading server can apply them consistently at scale.

\section{Red-Teaming the Assignment}
\label{sec:red-team}

A good AI-native assignment should be attacked before it is released. The red-team question is not ``Can anyone use AI?'' The answer is yes, and that is intended. The question is:
\begin{quote}
Can a low-skill AI-use strategy obtain above-baseline credit?
\end{quote}

The first attack is generic delegation: give the assignment, evaluator description, and public frontier to AI and ask it to beat the baseline. This attack should be part of baseline generation. If it works after the frontier is published, the frontier was not strong enough.

The second attack is protocol leakage: ask AI to infer or recover hidden inputs, hidden outputs, per-instance errors, or unauthorized evaluator information. This is not AI-use competence; it is an attack on the feedback oracle. The server design should prevent it by releasing only fixed aggregate feedback.

The third attack is frontier diagnosis. Here the student does not simply delegate the assignment; the student identifies a part of the frontier to improve, gives a task-specific reason why improvement may be possible, and asks AI to implement or test that idea. This is not a failure of the assignment. The prompt may itself contain the conceptual structure the course wants to teach.

\begin{definition}[Prompt-only process]
A \emph{prompt-only process} is an allowed AI-use process represented by an auditable trace of natural-language prompts, AI outputs, official server feedback, and candidate-selection decisions. The student's direct contribution in such a trace is prompting, selection among AI-generated candidates, and interpretation of official server feedback. The AI may write and revise code, but the student does not manually implement the final technical change except through prompts. If manual edits are allowed, they are treated as ordinary human-AI submissions rather than prompt-only submissions.
\end{definition}

\begin{definition}[Frontier-improvement rationale]
A \emph{frontier-improvement rationale} for task \(T\) is a task-specific explanation that identifies: (i) the frontier region or evaluator coordinate being targeted; (ii) a hypothesis for why the declared baseline may be improvable there; (iii) a proposed design change, algorithmic idea, representation, parameter choice, or solution strategy; and (iv) the expected tradeoff in the official evaluator coordinates. The rationale is valid when these claims are testable under \(E_T\) or by the supporting assessment artifacts. Validity here means testability and task relevance; it does not require that the hypothesis be correct before evaluation.
\end{definition}

\begin{definition}[Prompt-skill implication]
Let \(F_0\) be the declared AI-native frontier used for grading, and let
\(\widehat{\delta}_T(\pi;F_0)\) denote the observed Pareto surplus of a submission \(\pi\)
under the official evaluator. A prompt-only process \(P\) is called
\emph{\(\eta\)-grade-bearing} if its auditable trace produces an admissible submission \(\pi_P\) such that
\[
\widehat{\delta}_T(\pi_P;F_0)>\eta .
\]
Equivalently, under thresholded grading, \(\pi_P\) receives positive surplus credit above
the AI-baseline grade.

An assignment satisfies \emph{prompt-skill implication at threshold \(\eta\)} if every
\(\eta\)-grade-bearing prompt-only process falls into at least one of the following categories: it exposes an incomplete baseline, it uses hidden-input leakage, hidden-output leakage, unauthorized per-instance feedback, or another feedback-protocol violation, or its trace contains a valid frontier-improvement rationale.
\end{definition}

This definition is intentionally assignment-agnostic. In a systems assignment, the rationale may identify a performance bottleneck. In a theory-heavy or modeling-heavy task, it may identify a missing invariant, a weaker bound, a failure mode, a distributional assumption, or another task-specific reason the declared frontier can be moved.

For the Bloom-filter construction, a frontier-improvement rationale might say that the GB-scale baseline appears query-time limited by random memory probes, propose a blocked Bloom-filter layout with double hashing, and predict the false-positive-rate versus throughput tradeoff. A student who can write such a prompt has already understood important parts of the assignment.

\begin{designcriterion}[Prompt-skill red-team checklist]
\label{crit:prompt-skill}
Before release, the instructor should red-team prompt-only strategies at threshold \(\eta\) using three checks:
\begin{enumerate}
    \item include generic delegation prompts and looped generate--feedback--revise prompts in the public baseline-generation process;
    \item verify that the server oracle blocks hidden-input leakage, hidden-output leakage, and unauthorized per-instance feedback;
    \item inspect any grade-bearing prompt-only traces found during red-teaming and classify each as a baseline gap, a protocol violation, or a trace containing a valid frontier-improvement rationale.
\end{enumerate}
\end{designcriterion}

\noindent\textbf{Practical implication.} If ``solve this and beat the baseline'' works, harden the baseline. If a prompt wins by giving a testable reason for moving a specific part of the frontier, then writing that prompt is already part of the intended learning outcome. This is a release criterion and checklist, not a theorem independent of the preceding definition.

\section{Discussion, Limitations, and Deployment}
\label{sec:discussion}

\paragraph{From speculation to outcome-driven debate.}
The motivation for this paper is practical: AI-assisted technical work is changing quickly, and course design needs objects that can be analyzed. Personal exposure to frontier AI development motivates the proposal, but it is not offered as evidence. The paper instead asks instructors to specify outcomes: the task, evaluator, AI protocol, feedback oracle, frontier, budget assumptions, threshold, and grading rule. This moves debate from speculation, fear, and analogy toward concrete designs that can be criticized.

\paragraph{Understanding the evaluator is already learning.}
The approximate-membership construction teaches the topic partly by making the evaluation criterion explicit. To understand what it means to win, a student must understand why false negatives are forbidden, why false positives are allowed, how false-positive rate trades off against memory, why hash count and hash quality matter, and why the same mathematical structure behaves differently at MB and GB scale. Reading the evaluator and the frontier is not peripheral gaming. It is operational understanding of the data structure.

\paragraph{The bar for human learning is raised.}
In a traditional randomized-algorithms or data-structures course, knowing the definition, proof, and standard false-positive analysis of a Bloom filter may be the main target. In this assignment, that knowledge is closer to the entry point. To move the frontier, students may need to combine the standard theory with systems-level choices: cache-aware blocking, double hashing, vectorization, batching, memory alignment, construction throughput, parallelism, NUMA effects, and memory-bandwidth-aware layout. This is a feature of AI-native education: once AI can produce the standard solution, the human learning target moves toward diagnosis, integration, and scale-aware design. It is also a debatable cost, because the assignment may blur the boundary between the nominal topic and adjacent systems knowledge.

\paragraph{When the frontier does not move.}
The null-frontier outcome is one of the useful failure modes of the framework. If no student artifact moves the frontier, the course has learned something concrete: the assignment, as designed, did not separate students with the intended skill from the declared AI baseline. That conclusion should be interpreted carefully. Some skills remain worth teaching as foundations, vocabulary, debugging tools, or prerequisites for higher-level reasoning. Students may need to understand Bloom filters, hashing, or randomized algorithms even if AI can generate standard implementations. But in a senior AI-native assessment, foundational importance alone may not be enough for the main grade-bearing objective. If a skill does not produce measurable advantage over a person with AI but without that skill, then the instructor should consider revising the evaluator, strengthening the assignment, or moving the grade-bearing target to a higher-level capability built on top of that foundation.

\paragraph{AI-use skill is part of the outcome.}
A student who uses AI better should often do better. Better AI use means formulating sharper hypotheses, asking for targeted variants, checking generated claims against the evaluator, and combining model output with domain knowledge. The assignment therefore rewards both domain skill and AI-use competence. The server-mediated feedback design is meant to limit the degree to which this becomes primarily a contest of who can buy more private model calls.

\paragraph{Limits and standing assumptions.}
The formal claims are relative to the declared task, evaluator, AI baseline-generation process, feedback oracle, frontier-freeze time, private-budget envelope, threshold, and grading rule. They do not characterize every possible future AI system or unlimited private budget. Budget neutrality is implemented by controlling official benchmark feedback and thresholding small surplus; it is a design goal to red-team, not a metaphysical guarantee. In particular, the Bloom-filter construction should be read as a practical test case for calibrated budget-neutrality margins, not as a proof that approximate-membership tasks are intrinsically budget-neutral. Prompt-skill implication is also a red-team criterion: generic prompts belong in the baseline, low-skill prompts that obtain surplus expose a weakness, and task-specific frontier-improvement rationales may be evidence of learning.

Pareto surplus certifies submitted artifacts, not unaided authorship or internal competence by itself. The student-skill inference comes from the full assessment design: the hidden evaluator, equal feedback policy, design report, ablations, prompt traces, reproducibility checks, and any oral or written follow-up the instructor chooses to include. The framework is easiest to apply when the course can define executable tasks and multidimensional evaluators. It is less immediate for objectives such as communication, taste, conceptual explanation, or ethical reasoning. The approach also requires instructor effort: baseline generation, feedback-oracle design, evaluator calibration, red-team testing, and threshold selection. These costs are real, but they are also the work required to make AI-native assessment explicit rather than implicit.

The first deployment deliberately focuses on MB- and GB-scale benchmarks. Very large out-of-core or distributed benchmarks are scientifically interesting, but they are expensive to evaluate, difficult to equalize across student hardware or server queues, and too slow for a course feedback cycle.

\paragraph{Planned COMP 480/580 deployment.}
The first planned classroom deployment is COMP 480/580 at Rice University in Fall 2026. The intended preparation is to construct and harden the AI frontiers, freeze the server-feedback policy, calibrate the thresholds, and red-team generic prompts, protocol attacks, and frontier-diagnostic prompts before release. The purpose of sharing this draft now is to invite precise criticism before implementation: stronger baselines, better feedback policies, sharper thresholds, and alternative tasks.

\section*{AI Usage Disclosure}

The author used generative AI tools, including OpenAI ChatGPT, as interactive assistants during the preparation of this manuscript. These tools were used for brainstorming, rewriting, organization, LaTeX drafting, checking readability, and red-teaming the exposition and formal definitions. They were not used as authors, and no AI system is credited with authorship. The author reviewed, edited, and is solely responsible for all claims, definitions, proofs, citations, examples, and final text in the manuscript. No empirical results are reported in this position paper, and no data were generated or analyzed by AI for the purpose of making empirical claims.

\appendix

\section{Proofs}
\label{app:proofs}

\subsection{Proof of Theorem~\ref{thm:skill-measure}}
\label{app:proof-skill}
\begin{proof}
For a finite frontier \(F_{\AI}(T)\) and a point \(E_T(\pi)\) strictly above the fixed reference point \(r\), adding \(E_T(\pi)\) increases dominated hypervolume exactly when the axis-aligned box from \(r\) to \(E_T(\pi)\) contains positive measure not already covered by the boxes generated by \(F_{\AI}(T)\). If some frontier point weakly dominates \(E_T(\pi)\), then the box \([r,E_T(\pi)]\) is already contained in that frontier box, so the surplus is zero. Conversely, if no frontier point weakly dominates \(E_T(\pi)\), then for each \(f\in F_{\AI}(T)\) there is a coordinate \(j(f)\) with \(f_{j(f)}<E_T(\pi)_{j(f)}\). Since the frontier is finite and \(E_T(\pi)\) is strictly above \(r\), the positive gaps over these witnesses, together with the coordinate distances from \(r\) to \(E_T(\pi)\), have a positive minimum. Hence there is a small open rectangle adjacent to \(E_T(\pi)\) that lies inside \([r,E_T(\pi)]\) and outside all frontier boxes. That rectangle has positive measure, so the surplus is positive. Therefore
\[
\delta_T(\pi)>0
\quad\Longleftrightarrow\quad
E_T(\pi)\notin\Down(F_{\AI}(T)).
\]
For a finite assessment score set \(S_H^\Gamma(T;\sigma)\), the hypervolume expansion \(\Delta_T^\Gamma(\sigma)\) is positive if and only if at least one finite score in \(S_H^\Gamma(T;\sigma)\) adds positive hypervolume. By the equivalence above, this is exactly the existence of a score not dominated by the AI-native frontier.
\end{proof}

\subsection{Proof of Theorem~\ref{thm:pareto-grading}}
\label{app:proof-grading}
\begin{proof}
Let \(\pi\in\Pi_{\AI}(T)\). Then \(E_T(\pi)\in S_{\AI}(T)\). Since \(F_{\AI}(T)=\PF(S_{\AI}(T))\), every AI-native score vector is weakly dominated by some point on the AI-native Pareto frontier. Therefore \(E_T(\pi)\in\Down(F_{\AI}(T))\), and adding \(E_T(\pi)\) to the frontier does not increase hypervolume. Thus \(\delta_T(\pi)=0\), and \(G_T(\pi)=g_{\AI}+\phi(0)=g_{\AI}\).

Conversely, if \(G_T(\pi)>g_{\AI}\), then \(\phi(\delta_T(\pi))>0\). Since \(\phi(0)=0\) and \(\phi\) is nonnegative, this requires \(\delta_T(\pi)>0\). Therefore every above-baseline surplus grade has positive measured Pareto surplus.
\end{proof}

\subsection{Proof of Theorem~\ref{thm:thresholded}}
\label{app:proof-thresholded}
\begin{proof}
Let \(\pi\) be an AI-only submission generated by a protocol \(\mathcal A\) in the declared private envelope \(\mathfrak A_{\Priv}\). By the definition of \(\Pi^*(\mathcal A)\), this means \(\pi\in\Pi^*(\mathcal A)\), so \(E_T^*(\pi)\in S^*(\mathcal A)\). By \(\varepsilon_{\bud}\)-saturation, adding the entire set \(S^*(\mathcal A)\) to \(F_0\) increases hypervolume by at most \(\varepsilon_{\bud}\); adding the single point \(E_T^*(\pi)\) therefore has ideal surplus at most \(\varepsilon_{\bud}\):
\[
\delta_T^*(\pi;F_0)\leq \varepsilon_{\bud}.
\]
By the one-sided evaluator-error assumption,
\[
\widehat\delta_T(\pi;F_0)\leq \delta_T^*(\pi;F_0)+\varepsilon_{\evalerr}\leq \varepsilon_{\bud}+\varepsilon_{\evalerr}=\eta_T.
\]
Hence
\[
\pospart{\widehat\delta_T(\pi;F_0)-\eta_T}=0,
\]
and \(G_T^{\eta}(\pi)=g_{\AI}\). If the evaluator is randomized, this conclusion holds at the confidence level attached to the one-sided error bound.

Conversely, if \(G_T^{\eta}(\pi)>g_{\AI}\), then \(\widehat\delta_T(\pi;F_0)>\eta_T\). Whenever the one-sided evaluator-error bound holds,
\[
\delta_T^*(\pi;F_0)
\geq \widehat\delta_T(\pi;F_0)-\varepsilon_{\evalerr}
> \eta_T-\varepsilon_{\evalerr}
=\varepsilon_{\bud}.
\]
Thus above-baseline surplus credit implies ideal surplus beyond the declared private-AI-budget margin.
\end{proof}

\subsection{Proof of Proposition~\ref{prop:multi-frontier}}
\label{app:proof-multi}
\begin{proof}
For an AI-native suite submission, each scale component belongs to the corresponding AI-native submission set. By \Cref{thm:pareto-grading}, each component receives its scale baseline surplus grade, and thresholding cannot increase it. Therefore the weighted aggregate is the weighted aggregate of the scale baselines. Conversely, if the aggregate grade is strictly larger than the aggregate baseline, then at least one component with positive weight \(w_\lambda>0\) must be strictly larger than its scale baseline, which implies threshold-exceeding Pareto surplus in that scale.
\end{proof}

\section{Assignment Handout Skeleton}

\paragraph{Task.}
You are given two approximate-membership tasks centered on Bloom filters: MB and GB. For each task, your program must build a compact representation of a set and answer membership queries with no false negatives. False positives are allowed and measured. The handout states whether submissions are restricted to Bloom-filter-family designs or may use broader approximate-membership systems.

\paragraph{Baseline.}
For each scale, we provide an AI-native Pareto frontier generated before the assignment using a fixed AI-agent protocol, including looped generate--feedback--revise baselines. The baseline-freeze record includes models, prompts, tools, budgets, server-feedback access, candidate counts, selection rules, seeds where applicable, and freeze time. Grade-bearing improvement requires exceeding the frontier by more than the published threshold for evaluator noise and plausible private AI-only budget advantage.

\paragraph{Submission.}
You may submit either one configurable implementation or separate implementations for the two scales. Your code must expose a fixed build/query interface and run on the specified hardware.

\paragraph{Correctness.}
Any false negative on the hidden benchmark is a hard failure for the corresponding scale. False positives are allowed and become part of the score vector.

\paragraph{Server feedback.}
The hidden benchmark inputs and outputs are not released. During development, the server returns only fixed aggregate feedback and enforces the same query budget for all students. You may use AI freely, and learning to use AI effectively is part of the assignment.

\paragraph{Collaboration and leakage.}
You may discuss course-permitted concepts and high-level design ideas according to the collaboration policy. Do not share final code, private prompts, tuned parameters, hidden-feedback summaries, or frontier-moving optimizations outside the allowed collaboration structure. Because the surplus component may be percentile-calibrated, giving another team your best solution can reduce your own relative grade.

\paragraph{Scoring.}
For each scale, your surplus score is determined by thresholded Pareto surplus over the published AI frontier. The course staff will not attempt to infer whether your code was written with AI assistance. The evaluator measures behavior, not authorship. Other course components may check correctness, reproducibility, report quality, and minimum-performance requirements.

\paragraph{Design report.}
Your report should explain which frontier you attempted to move, what tradeoff you targeted, which design ideas you tried, what failed, and why your final implementation improves or fails to improve the baseline. Supporting evidence may include ablations, prompt traces, profiling data, and reproducibility notes.

\section{Notation Summary}

\begin{longtable}{p{0.24\linewidth}p{0.68\linewidth}}
\toprule
Symbol & Meaning \\
\midrule
\(T\) & Real task \\
\(\mathcal I\) & Task interface and benchmark protocol \\
\(\Omega_{\mathcal I}\) & Raw submission space induced by \(\mathcal I\) \\
\(\Pi_{\mathcal I}\) & Admissible submissions induced by \(\mathcal I\) \\
\(E_T\) & Total executable evaluator mapping raw submissions to score vectors, with frontier claims restricted to admissible nonfailing scores \\
\(S_{\AI}(T)\) & Score set of AI-native submissions under the declared baseline process \\
\(E_T^*,\widehat E_T\) & Ideal and observed evaluator scores \\
\(F_{\AI}(T)\) & AI-native Pareto frontier \\
\(\delta_T(\pi)\) & Pareto surplus of submission \(\pi\) over the AI frontier \\
\(\sigma\) & Candidate human skill \\
\(\Gamma\) & Finite assessment protocol used to produce human or student score vectors \\
\(\Delta_T^\Gamma(\sigma)\) & Finite-protocol hypervolume expansion associated with candidate skill \(\sigma\) \\
\(G_T\) & Grading rule \\
\(\mathcal A\) & AI-agent protocol \\
\(\mathcal O_{\srv}\) & Server-mediated feedback oracle \((\mathcal B_{\hid},q,E_{T,\hid},\psi)\) \\
\(\psi\) & Fixed feedback map hiding inputs, outputs, and per-instance errors \\
\(\mathfrak A_{\Priv}\) & Declared envelope of plausible private AI-only protocols \\
\(\eta_T\) & Budget-and-noise threshold used before awarding surplus credit \\
\(F_0\) & Public frontier used for thresholded grading \\
\(T_{\MB},T_{\GB}\) & MB- and GB-scale approximate-membership tasks centered on Bloom filters \\
\(\widehat F_{\AI}^{\lambda}\) & Empirical AI frontier for scale \(\lambda\) \\
\bottomrule
\end{longtable}

\end{document}